\documentclass[11pt]{article}
\usepackage{amsmath,amsthm,amssymb}
\newcommand{\remove}[1]{}
\newtheorem{thm}{Theorem}[section]
\newtheorem{claim}[thm]{Claim}
\newtheorem{lem}[thm]{Lemma}
\newtheorem{define}[thm]{Definition}
\newtheorem{cor}[thm]{Corollary}

\newtheorem{THM}{Theorem}

\def\F{{\mathbb{F}}}
\def\bF{{\overline{\mathbb{F}}}}

\def\C{{\mathbb{C}}}

\def\V{{\mathbf{V}}}

\def\E{{\mathbb E}}

\newcommand{\ip}[2]{\langle #1,#2 \rangle}

\def\_{\,\,\,\,\,}

\def\poly{\textsf{poly}}

\def\Tr{\mathrm{Tr}}
\def\ep{\mathrm{e}_p}

\newcommand{\eps}{\epsilon}

\begin{document}

\title{Subspace Evasive Sets}
\author{
Zeev Dvir\thanks{Department of Computer Science, Princeton University, Princeton NJ.
Email: \texttt{zeev.dvir@gmail.com}. Research partially
supported by NSF grant CCF-0832797 and by the Packard fellowship.}
\and
Shachar Lovett
\thanks{School of Mathematics, Institute for Advanced Study, Princeton, NJ. Email:
\texttt{slovett@math.ias.edu}. Research supported by NSF grant DMS-0835373.}\\
}
%Wigderson\footnotemark[1]}
\date{}
\maketitle

\begin{abstract}
 In this work we describe an explicit, simple, construction of large subsets of $\F^n$, where $\F$ is a finite field, that have small intersection with every $k$-dimensional affine subspace. Interest in the explicit construction of such sets, termed  {\em subspace-evasive} sets, started in the work of Pudl\'{a}k and R\"{o}dl \cite{PudlakRodl04} who showed how such constructions over the binary field can be used to construct explicit Ramsey graphs. More  recently, Guruswami \cite{Guruswami11}  showed that, over large finite fields (of size polynomial in $n$), subspace evasive sets can be used to obtain explicit list-decodable codes with optimal rate and constant list-size. In this work we construct subspace evasive sets over large fields and use them, as described in \cite{Guruswami11}, to reduce the list size of folded Reed-Solomon codes form $\poly(n)$ to a constant.
\end{abstract}

%\newpage
%\pagenumbering{arabic}

%%%%%%%%%%%%%%%%%%%%%%%%%%%%%%%%%%%%%%%%%%%%%%%%%%%%%%%%%%%%
\section{Introduction}
%%%%%%%%%%%%%%%%%%%%%%%%%%%%%%%%%%%%%%%%%%%%%%%%%%%%%%%%%%%%

\subsection{Subspace evasive sets} % (fold)
\label{sub:subspace_evasive_sets}

% subsection subspace_evasive_sets (end)

Defined formally, a $(k,c)$-subspace evasive set $S \subset \F^n$ has intersection of size at most $c$ with every $k$-dimensional affine subspace $H \subset \F^n$. This definition makes sense over finite fields, as well as over infinite fields. Over finite fields, a simple probabilistic argument shows that a random set $S$ of size $|\F|^{(1-\eps)n}$ will have intersection of size at most $c(k,\eps) = O(k/\eps)$ with any $k$-dimensional affine subspace $H$. In this work we give the first explicit construction of a subspace-evasive set $S$ of size $|\F|^{(1-\eps)n}$ that has intersection size at most $c(k,\eps) = (k/\eps)^{k}$ with every $k$-dimensional affine subspace $H$. This is stated in the next theorem.
We postpone the exact definition of the term {\em explicit} to the following sections (see Theorem~\ref{thm-evasive} for the formal statement of this theorem and Section~\ref{sec-explicit} for a discussion of explicitness).

\begin{THM}[Main theorem]\label{THM-main}
For any finite field $\F$ and parameters $k \ge 1,\eps>0$ there exists an explicit construction of a set $S \subset \F^n$ of size $|S| > |\F|^{(1-\eps)n}$ that is $(k,c(k,\eps))$-subspace evasive with $c(k,\eps) = (k/\eps)^k$.
\end{THM}

While  being far from the optimal bound of $O(k/\eps)$ and despite being exponential in $k$, the bound we obtain is  useful when $k$ is small and the field is sufficiently large. As we will see below, this is precisely the setting that was raised by Guruswami in connection to error correcting codes.

The main ingredient in our construction is an explicit family of degree $d$ polynomials $f_1,\ldots,f_k \in \F[x_1,\ldots,x_n]$, for all $k \leq n \leq |\F|$, such that for every injective (i.e full rank) affine map $\ell : \F^k \mapsto \F^n$ the system of equations
\begin{eqnarray*}
&f_1(\ell(t_1,\ldots,t_k))=0 \\
&\vdots \\
&f_k(\ell(t_1,\ldots,t_k))=0.
\end{eqnarray*}
has at most $d^k$ solutions. The degree  $d$ can be any number  between $n$ and $|\F|$. Using algebraic-geometry terminology, the set of common zeros of $f_1,\ldots,f_k$ forms an $(n-k)$-dimensional variety which has finite intersection with any $k$ dimensional affine sub-space. We call such varieties {\em everywhere-finite} varieties (see Section~\ref{sec-equations} for a longer discussion of this particular choice of name).

Constructing subspace evasive sets as in Theorem~\ref{THM-main} is then obtained by partitioning the $n$ coordinates of the space into blocks of size $k/\eps$ and applying the basic construction (of an everywhere-finite variety) on each block independently. The polynomials we use in the basic construction are extremely simple (weighted sums of powers of variables) which makes the final construction explicit enough to be useful for the list-decoding application described in \cite{Guruswami11} (allowing for both efficient encoding {\em and} list-decoding). Our proofs are elementary and do not use any sophisticated algebraic machinery (apart from Bezout's theorem).\footnote{We do use Weil's exponential sum estimates to analyze a certain variant of the construction but this part of the proof can be omitted by choosing the polynomials $f_1,\ldots,f_k$ more carefully (as described in Section~\ref{sec-explicit}).}

\subsection{List-decodable codes} % (fold)
\label{sub:list_decodable_codes}
An error-correcting code allows one to encode a message into a codeword so that encodings of different messages differ in many coordinates. This allows one to  recover the original message from an encoding that is corrupted in a small number of coordinates. More formally, A code is a subset  $C  \subset \Sigma^m$, where $\Sigma$ is some finite alphabet. The {\em rate} of the code is denoted $R = \frac{\log|C|}{m\log|\Sigma|}$ and the {\em distance} of the code, denoted $\rho$, is the minimal Hamming distance between two codewords divided by $m$. It is easy to show that $\rho < 1- R$ and that {\em unique decoding} (i.e decoding a message uniquely from a corrupted codeword) is only possible from a fraction $(1-R)/2$ of errors. When the number of errors goes beyond $(1-R)/2$ one has to be satisfied with {\em list-decoding}, in which a  short list of possible  messages is returned (i.e all messages whose encodings are close to the received word). Non explicitly, one can show the existence of a code that can be list-decoded from $1-R-\eps$ errors with list-size bounded by $O(1/\eps)$. Obtaining an explicit construction of such a code (with efficient encoding/decoding) is a major open problem in coding theory. The first work to give explicit codes that can be list-decoded from $1-R-\eps$ errors was the paper of Guruswami and Rudra \cite{GuruswamiRudra08} which builds on earlier work by Parvaresh and Vardi \cite{PV05}. Their work showed that a certain family of codes, called {\em folded Reed-Solomon (RS) codes} can be list-decoded from $1- R - \eps$ errors with list size bounded by $m^{O(1/\eps)}$, where $m$ is the number of coordinates (or block length)
 of the code.

In a recent work, Guruswami \cite{Guruswami11} gave a new list-decoding algorithm for folded RS codes which have some nice advantages over previous decoding algorithms. Among these advantages is the property that the list of possible messages, returned by the decoder, is contained in a low dimensional subspace. More precisely, the code represents messages  as elements of $\F^n$, where $\F$ is a finite field of size $q \sim n$, and the list returned by the decoder is (quite surprisingly) a {\em subspace} of dimension $O(1/\eps)$. This immediately gives the size bound for the list of $q^{O(1/\eps)}$ mentioned above but also shows a way for improving further the list size. Guruswami observed that restricting the messages to come from a $((1/\eps), c(\eps))$-evasive set $S \subset \F^n$, instead of coming from the entire space $\F^n$, will reduce the list size to $c(\eps)$ and remove the dependency on the block length. In order for the rate to not degrade by much we need the size of $S$ to be sufficiently large, say $|S| > |\F|^{(1-\eps)n}$.

For this application to produce codes with efficient encoding/decoding, the evasive set $S$ must satisfy two explicitness conditions. The first is that messages can be encoded and decoded efficiently into $S$. The second condition is that, given a subspace (say, as a list of basis vectors), one can efficiently compute the intersection of this subspace with $S$. Our construction of subspace evasive sets satisfies both of these conditions  (see Section~\ref{sec-explicit}) and so we obtain the following theorem.
\begin{THM}
\label{thm:list-decodable-codes}
For every $R$ and $\eps$, there exists an explicit family of codes $C \subset \Sigma^m$ with rate $R$ that can be list-decoded from a fraction $1 - R - \eps$ of errors in quadratic time and with list size $(1/\eps)^{O(1/\eps)}$.
\end{THM}

The use of evasive sets to enhance list-decoding is completely black-box and only uses the property that the returned list is a subspace of a certain dimension in a sufficiently large field. We give the proof of Theorem~\ref{thm:list-decodable-codes} in Section~\ref{sec-listdec}, stating the relevant claims from \cite{Guruswami11} that are needed for the black-box application.

Following \cite{Guruswami11}, Guruswami and Wang \cite{GuruswamiWang11} showed another family of codes with optimal distance list decoding and with the additional property that the list returned by the decoder is a subspace. This family of codes, called derivative codes (also called multiplicity codes in \cite{KSY11}), obtains roughly the same parameters as folded RS codes and can be also combined with our construction of evasive sets in the same way to reduce the list size.
% subsection list_decodable_codes (end)

%\paragraph{Notation}
%We denote $[n]=\{1,\ldots,n\}$. Let $\F$ be a finite field. Its algebraic closure is denoted by $\bF$. We denote elements in $\F^n$ or $\bF^n$ by $x=(x_1,\ldots,x_n)$.

\subsection{Affine and two-source extractors} % (fold)
\label{sub:affine_and_two_source_extractors}
The work of Pudl\'{a}k and R\"{o}dl~\cite{PudlakRodl04} showed that constructing $(n/2,c)$-subspace evasive sets $S \subset \F_2^n$ gives explicit constructions of bipartite Ramsey graphs. These are bipartite graphs that do not contain bipartite cliques or independent sets of certain size. A recent work of Ben-Sasson and Zewi \cite{BenSassonZewi11} explored this connection further and showed (under some number theoretic conjectures) that such sets can also be used to construct two-source {\em extractors} which are strong variants of bipartite Ramsey graphs. Another application given in \cite{BenSassonZewi11} was to the construction of {\em affine extractors} which are functions that have uniform output whenever the input chosen uniformly from a subspace of sufficiently high dimension. Both of these applications require that the construction be over a field of two elements. Our construction requires the field to be at least of size $n$ and so is not useful for these applications. An important direction for progress is to generalize our construction for smaller fields. Alternatively, one can try to generalize the approach of \cite{BenSassonZewi11} to larger fields and then try to use our construction to obtain better extractors (affine or two-source).
% subsection affine_and_two_source_extractors (end)

\subsection{Organization} % (fold)
\label{sub:organization}

% subsection organization (end)
Section~\ref{sec-equations} contains the main construction of everywhere-finite varieties (Theorem~\ref{thm-explicit-efv}). In Section~\ref{sec-evasive}, we show how to compose this basic construction to obtain our  main theorem, Theorem~\ref{thm-evasive}, which gives explicit evasive sets. In Section~\ref{sec-explicit} we prove several claims which deal with the explicitness of our construction, and use them, in  Section~\ref{sec-listdec} to derive Theorem~\ref{thm:list-decodable-codes}. Appendix~\ref{sec:preliminaries} contains some basic results on Fourier analysis that are used in part of Section~\ref{sec-evasive}.

%%%%%%%%%%%%%%%%%%%%%%%%%%%%%%%%%%%%%%%%%%%%%%%%%%%%%%%%%%%%
\section{Everywhere-finite varieties}\label{sec-equations}
%%%%%%%%%%%%%%%%%%%%%%%%%%%%%%%%%%%%%%%%%%%%%%%%%%%%%%%%%%%%

Let $\F$ be a field and $\bF$ its algebraic closure (recall that the algebraic closure is always infinite, even if $\F$ is finite). A variety in $\bF^n$ is the set of common zeros of one or more polynomials. Given $k$ polynomials $f_1,\ldots,f_k \in \bF[x_1,\ldots,x_n]$, we denote the variety they define as
$$
\V(f_1,\ldots,f_k) := \{ x \in \bF^n \,|\, f_1(x) = \ldots = f_k(x) = 0\}.
$$
The {\em dimension} of a variety is a generalization of the notion of dimension for subspaces and can be thought of, informally, as the number of `degrees of freedom' the variety has. In particular, $k$ generic polynomials  $f_1,\ldots,f_k$ define a variety  $\V(f_1,\ldots,f_k)$ of dimension $n-k$. It is well known that the intersection of an $(n-k)$-dimensional variety $V \subset \bF^n$ with a generic $k$ dimensional affine subspace $H \subset \bF^n$ is finite\footnote{For a precise definition of dimension and proofs of its basic properties we refer the reader to any elementary text on Algebraic Geometry (e.g \cite{Shaf}).}. In the following we will not rely on any of these properties and keep the discussion self-contained. Our main result in this section is a construction of an explicit variety $V$ where this holds for {\em all} affine subspaces $H$ of dimension $k$. Using Bezout's theorem (Theorem~\ref{thm-bezout}) and the bound on the degrees of the polynomials defining  $V$ we will also get an explicit uniform bound on the size of the intersections $|\V \cap H|$. We start with the formal definition.

\begin{define}[Everywhere-finite variety]
Let $f_1,\ldots,f_k \in \bF[x_1,\ldots,x_n]$ be polynomials. The variety $\V=\V(f_1,\ldots,f_k)$ is {\em everywhere-finite} if for any affine subspace $H \subset \bF^n$ of dimension $k$, the intersection $\V \cap H$ is finite.
\end{define}

The importance of showing that the intersection is finite comes from Bezout's theorem, which allows one to give explicit bounds on the intersection size, given that it is finite. This result can be found in most introductory texts on Algebraic Geometry \cite{Shaf} (for an elementary proof of this particular formulation see \cite{Sch95}).

\begin{thm}[Bezout]\label{thm-bezout}
Let $g_1,\ldots,g_k \in \bF[t_1,\ldots,t_k]$ be polynomials. If $\V(g_1,\ldots,g_k)$ is finite then $$|\V(g_1,\ldots,g_k)| \le \prod_{i=1}^k \deg(g_i).$$
\end{thm}

For everywhere-finite varieties this gives the following immediate corollary.

\begin{cor}\label{cor-size-bound-bezout}
Let $f_1,\ldots,f_k \in \bF[x_1,\ldots,x_n]$ be polynomials such that $\V=\V(f_1,\ldots,f_k)$ is everywhere-finite. Then for any $k$-dimensional affine subspace $H \subset \bF^n$ we have$$
|\V \cap H| \le \prod_{i=1}^k \deg(f_i).
$$
\end{cor}

\begin{proof}
Let the $k$-dimensional affine subspace $H$ be given as the image of an affine map $\ell=(\ell_1,\ldots,\ell_n):\bF^k \to \bF^n$. Let $g_i \in \bF[t_1,\ldots,t_k]$ denote the restriction of $f_i$ to $H$, i.e.
$$
g_i(t_1,\ldots,t_k) := f_i(\ell_1(t_1,\ldots,t_k),\ldots,\ell_n(t_1,\ldots,t_k)).
$$
Clearly $\V \cap H = \V(g_1,\ldots,g_k)$ and $\deg(g_i) \leq \deg(f_i)$. The corollary now follows from Theorem~\ref{thm-bezout}.
\end{proof}

We will now describe an explicit construction of an everywhere-finite variety. We will need the following definition:
A $k \times n$ matrix (where $k \le n$) is {\em $k$-regular} if all its $k \times k$ minors are regular (i.e have non-zero determinant).
For example, if $\F$ is a field with at least $n$ distinct nonzero elements $\gamma_1,\ldots,\gamma_n$ then
$A_{i,j}=\gamma_j^i$ is $k$-regular.

\begin{thm}[Construction of an everywhere-finite variety]\label{thm-explicit-efv}
Let $1 \le k \le n$ be parameters and $\F$ be a field.
Let $A$ be a $k \times n$ matrix with coefficients in $\F$ which is $k$-regular. Let $d_1>d_2>\ldots>d_n \ge 1$ be integers.
Let the polynomials $f_1,\ldots,f_k \in \F[x_1,\ldots,x_n]$ be defined as follows:
$$
f_i(x_1,\ldots,x_n) := \sum_{j=1}^n A_{i,j} \cdot x_j^{d_j}.
$$
Then $\V=\V(f_1,\ldots,f_k)$ is everywhere-finite. In particular, for any $k$-dimensional affine subspace $H \subset \bF^n$ we have $|\V \cap H| \le (d_1)^k$.
\end{thm}

We prove Theorem~\ref{thm-explicit-efv} in the remainder of this section. Let $H \subset \bF^n$ be a $k$-dimensional affine subspace. Our goal is to show that $\V \cap H$ is finite, and then the size bound follows from Corollary~\ref{cor-size-bound-bezout}. The first step is to present $H$ as the image of an affine map $\ell : \bF^k \mapsto \bF^n$ with a convenient choice of basis. In the following let $t=(t_1,\ldots,t_k) \in \bF^k$ and $x=(x_1,\ldots,x_n) \in \bF^n$.

\begin{claim}\label{cla-ell}
There exists an affine map $\ell=(\ell_1,\ldots,\ell_n):\bF^k \to \bF^n$ whose image is $H$ and such that the following holds. There exist $k$ indices $1 \leq j_1 < j_2 < \ldots < j_k \leq n$ such that
\begin{enumerate}
\item For all $i \in [k]$, $\ell_{j_i}(t) = t_i$.
\item If $j < j_1$ then $\ell_{j}(t) \in \bF$ (i.e $\ell_j$ is constant).
\item If $j< j_{i}$ for $i>1$ then $\ell_j(t)$ is an affine function just of the variables $t_1,t_2,\ldots,t_{i-1}$.
\end{enumerate}
\end{claim}

\begin{proof}
Let $\ell':\bF^k \to \bF^n$ be an arbitrary affine map whose image is $H$. We construct $\ell$ by a basis change of $\ell'$ which puts it in an upper-echelon form. That is, let $j_1$ be the minimal index such that $\ell'_{j_1}(t)$ is not constant. We take $\ell_{j_1}(t)=t_1$. Let $j_2$ be the minimal index after $j_1$ such that $\ell'_{j_2}(t)$ is not an affine function of $\ell'_{j_1}(t)$. We take $\ell_{j_2}(t)=t_2$, and we have that $\ell_{j}(t)$ for $j_1<j<j_2$ are affine functions of $t_1$. Generally, let $j_i$ be the minimal index after $j_{i-1}$ such that $\ell'_{j_i}(t)$ is not an affine function of $\ell'_{j_1}(t),\ldots,\ell'_{j_{i-1}}(t)$. We take $\ell_{j_i}(t)=t_i$ and have that $\ell_j(t)$ for $j_{i-1}<j<j_i$ are affine functions of $t_1,\ldots,t_{i-1}$. Obviously, for $j>j_k$ we have that $\ell_j(t)$ are affine functions of all $t_1,\ldots,t_k$.
\end{proof}

Let $\ell = (\ell_1,\ldots,\ell_n):\bF^k \to \bF^n$ be given by Claim~\ref{cla-ell} and let $1 \leq j_1 < j_2 < \ldots < j_k \leq n$ be the indices given by the claim. Let
$J  := \{j_1,\ldots,j_k\}$. Our goal is to show that the following system has a finite number of solutions:
$$
f_i(\ell_1(t_1,\ldots,t_k),\ldots,\ell_n(t_1,\ldots,t_k)) = 0, \,\, i \in [k].
$$

Clearly, applying an invertible linear transformation on the set $f_1,\ldots,f_k$ (replacing each $f_i$ with a linear combination of $f_1,\ldots,f_k$) will not affect the number of solutions. Our next step is to find such a linear transformation that will put the $f_i$'s in a more convenient form, eliminating some of their coefficients.
\begin{claim}\label{cla-linearf}
Let $f(x) = (f_1(x),\ldots,f_k(x))$. There exist $k$ linearly independent vectors $u_1,\ldots,u_k \in \bF^k$ such that, for all $i \in [k]$,
\begin{equation}\label{eq-tildef}
\ip{u_i}{f(x)} = x_{j_i}^{d_{j_i}} + \sum_{j \in [n] \setminus J} c_{ij} \cdot x_j^{d_j},
\end{equation} where the coefficients $c_{ij}$ are elements of $\bF$.
\end{claim}

\begin{proof}
Recall that by definition $f_i(x)=\sum_{j=1}^n A_{i,j} \cdot x_j^{d_j}$ where $A$ is a $k$-regular matrix. Let $A'$ be the $k \times k$ minor of $A$ given by restriction to columns $j_1,\ldots,j_k$. Since $A$ is $k$-regular we have that $A'$ is regular.
Let $u_1,\ldots,u_k \in \F^k$ denote the rows of $(A')^{-1}$. We thus have that $u_i A' = e_i$ where $e_i$ is the $i$-th unit vector. That is, $\ip{u_i}{f(x)} = x_{j_i}^{d_{j_i}} + \sum_{j \notin J} c_{ij} \cdot x_j^{d_j}$ where
$c_{ij}$ is the inner product of $u_i$ and the $j$-th column of $A$.
\end{proof}

Let $u_1,\ldots,u_k$ be the vectors given by Claim~\ref{cla-linearf} and denote
$$
\tilde f_i(x)  := \ip{u_i}{f(x)}.
$$
Let us also denote
$$
g_i(t_1,\ldots,t_k) := \tilde f_i(\ell_1(t_1,\ldots,t_k),\ldots,\ell_n(t_1,\ldots,t_k)).
$$
Recall that, from the above discussion, our goal is to show that the system $\{g_i(t)=0: i \in [k]\}$ has a finite number of solutions in $\bF^k$. By Claims~\ref{cla-ell} and~\ref{cla-linearf} we have that
\begin{equation}\label{eq-gi}
g_i(t) = t_i^{d_{j_i}} + \sum_{j \in [n] \setminus J} c_{ij} \cdot \ell_j(t)^{d_j}.
\end{equation}
We now perform one final transformation on our system. Contrary to the previous transformations which were linear transformations, this will be a polynomial transformation. Let
$$
D := \prod_{i=1}^k d_{j_i}
$$
and let
$$
D_i := \frac{D}{d_{j_i}}.
$$
For $i \in [k]$ define
$$
h_i(t_1,\ldots,t_k) := g_i\left( t_1^{D_1},\ldots, t_k^{D_k} \right).
$$
We first note that in order to show that $\V(g_1,\ldots,g_k)$ is finite it suffices to show that $\V(h_1,\ldots,h_k)$ is finite.

\begin{claim}\label{cla-g-finite-if-h-finite}
$|\V(g_1,\ldots,g_k)| \le |\V(h_1,\ldots,h_k)|$.
\end{claim}

\begin{proof}
For each $w \in \V(g)$ we can define $w' \in \V(h)$ by letting $w'_i$ be some $D_i$ root of $w_i$ (it exists
since $\bF$ is algebraically closed). Clearly distinct elements in $\V(g)$ are mapped to distinct elements in $\V(h)$.
\end{proof}

The reason for these transformations is that the final polynomials $h_i$ have a specifically nice form: they are the sum of $t_i^{D}$ with a polynomial of lower total degree.

\begin{claim}\label{cla-degR}
For all $i \in [k]$ we have that
$$
h_i(t_1,\ldots,t_k)=t_i^D+r_i(t_1,\ldots,t_k)
$$
where $\deg(r_i)<D$.
\end{claim}
\begin{proof}
By definition
$$
h_i(t)=g_i(t_1^{D_1},\ldots,t_k^{D_k})=t_i^D + \sum_{j \in [n] \setminus J} c_{ij} \cdot \ell_j(t_1^{D_1},\ldots,t_k^{D_k})^{d_j}.
$$
To prove the claim we need to show that $\deg(\ell_j(t_1^{D_1},\ldots,t_k^{D_k}))<D/d_j$ for all $j \notin J$. If $j<j_1$ then $\ell_j$ is constant. Otherwise let $i \in [k]$ be maximal such that $j>j_i$. By Claim~\ref{cla-ell} we have that
$\ell_j(t)$ is an affine function of $t_1,\ldots,t_i$. Since $D_1< \ldots < D_k$ we have that
$$
\deg(\ell_j(t_1^{D_1},\ldots,t_k^{D_k})) \le D_{i} = \frac{D}{d_{j_i}} < \frac{D}{d_j}
$$
since $d_1>\ldots>d_n$.
\end{proof}

To complete the proof of Theorem~\ref{thm-explicit-efv} we need to show that $\V(h_1,\ldots,h_k)$ is finite. This follows from a general bound for polynomials of the form $h_i(t)=t_i^D + r_i(t)$ where $\deg(r_i)<D$.

\begin{lem}\label{lem-reduce}
Let $h_1,\ldots,h_k \in \bF[t_1,\ldots,t_k]$ be polynomials such that $h_i(t)=t_i^D + r_i(t)$ where $\deg(r_i)<D$. Then $\V(h_1,\ldots,h_k) \le D^k$.
\end{lem}

Lemma~\ref{lem-reduce} follows immediately from the following two claims. In the following, let $R:=\bF[t_1,\ldots,t_k]$ be the ring of polynomials; $I:=\langle h_1,\ldots,h_k \rangle$ be the ideal in $R$ generated by $h_1,\ldots,h_k$; and $M:=R/I$ be their quotient. Note that $M$ is a vector space over $\bF$.

\begin{claim}
$
|\V(h_1,\ldots,h_k)| \le \dim M.
$
\end{claim}

\begin{proof}
Assume by contradiction there exist $w_1,\ldots,w_m \in \V(h_1,\ldots,h_k)$ where $m>\dim(M)$. Let $q_i \in \bF[t_1,\ldots,t_k]$ be polynomials
such that $q_i(w_i)=1$ and $q_i(w_j)=0$ for all $j \ne i$. Let $\tilde{q}_i$ be the image of $q_i$ in $M$. Since $m>\dim(M)$ there must exist a nonzero linear dependency among $\tilde{q}_1,\ldots,\tilde{q}_m$. That is, there exist $c_1,\ldots,c_m \in \bF$ not all zero such that
$$
\sum_{i \in [m]} c_i \cdot \tilde{q}_i(t) = 0 \quad (\textrm{in } M).
$$
Equivalently put,
$$
\sum_{i \in [m]} c_i \cdot q_i(t) \in I.
$$
The key observation is that for any polynomial $h(t) \in I$ we have that $h(w)=0$ for all $w \in \{w_1,\ldots,w_m\}$. This is because $h_i(w)=0$ for all $i \in [k]$ by assumption. Thus substituting $t=w_j$ we get that
$$
0=\sum_{i \in [m]} c_i \cdot q_i(w_j) = c_j,
$$
which contradicts the assumption that not all $c_1,\ldots,c_m$ are nonzero.
\end{proof}

\begin{claim}
$\dim M \le D^k$.
\end{claim}

\begin{proof}
We will show that $M$ is spanned by the image in $I$ of the monomials $t_1^{e_1} \ldots t_k^{e_k}$ where $0 \le e_1,\ldots,e_k \le D-1$. Thus in particular $\dim M \le D^k$. In order to do so, we need to show that if $q(t)$ is a polynomial then there exists a polynomial $\tilde{q}(t)$ such that $q-\tilde{q} \in I$ and the degree of each variable in $\tilde{q}$ is at most $D-1$. It suffices to show that if $q(t)$ has some variable of degree at least $D$ then we can find $\tilde{q}$ such that $q-\tilde{q} \in I$ and such that $\deg(\tilde{q})<\deg(q)$. The claim then follows by iterating this process until all variables have degrees below $D$. Moreover, it suffices to prove this in the case where $q$ is a monomial, as this process can be applied to each monomial individually.

Thus, let $q(t)=t_1^{e_1} \ldots t_k^{e_k}$ be a monomial where $e_i \ge D$ for some $i \in [k]$. Define
$$
\tilde{q}(t) = t_i^{e_i-D} (t_i^D-h_i(t)) \cdot \prod_{j \ne i} t_j^{e_j}.
$$
We have that $\deg(\tilde{q})<\deg(q)$ since $\deg(h_i(t)-t_i^D)<D$ by assumption; and $q(t)-\tilde{q}(t)=h_i(t) t_i^{e_i-D} \prod_{j \ne i} t_j^{e_j} \in I$ as required.
\end{proof}

%%%%%%%%%%%%%%%%%%%%%%%%%%%%%%%%%%%%%%%%%%%%%%%%%%%%%%%%%%%%%%%%%%%%%%%%%%
\section{Subspace Evasive sets}\label{sec-evasive}
%%%%%%%%%%%%%%%%%%%%%%%%%%%%%%%%%%%%%%%%%%%%%%%%%%%%%%%%%%%%%%%%%%%%%%%
In this section we construct subspace evasive sets, based on the construction of everywhere-finite varieties given in Theorem~\ref{thm-explicit-efv}. We first recall the definition of subspace evasive sets.

\begin{define}[Subspace evasive sets]
Let $S \subset \F^n$. We say $S$ is $(k,c)$-subspace evasive if for all $k$-dimensional affine subspaces $H \subset \F^n$ we have $|S \cap H| \le c$.
\end{define}

We next give some necessary definitions. For polynomials $f_1,\ldots,f_k \in \F[x_1,\ldots,x_m]$ we define their common solutions in $\F^m$ (as opposed to their solutions over the algebraic closure) as
$$
\V_{\F}(f_1,\ldots,f_k) := \V(f_1,\ldots,f_k) \cap \F^m = \{x \in \F^m: f_1(x)=\ldots=f_k(x)=0\}.
$$
We say that a $k \times m$ matrix is {\em strongly-regular} if all its $r \times r$ minors are regular for all $1 \le r \le k$. For example, if $\F$ is a field with at least $m$ distinct nonzero elements $\gamma_1,\ldots,\gamma_m$ then $A_{i,j}=\gamma_j^i$ is strongly-regular.

\begin{thm}\label{thm-evasive}
Let $k \ge 1, \eps>0$ and $\F$ be a finite field. Let $m:=k/\eps$ and assume  $m$ is integer and $m$ divides $ n$. Let $A$ be a $k \times m$ matrix with coefficients in $\F$ which is strongly-regular. Let $d_1>\ldots>d_m$ be integers. For $i \in [k]$ let
$$
f_i(x_1,\ldots,x_m) := \sum_{j=1}^m A_{i,j} \cdot x_j^{d_j},
$$
and define $S \subset \F^n$ to be the $(n/m)$-times cartesian product of $\V_{\F}(f_1,\ldots,f_k) \subset \F^m$. That is
\begin{eqnarray*}
S &=& \V_{\F}(f_1,\ldots,f_k) \times \ldots \times \V_{\F}(f_1,\ldots,f_k)\\
&=& \{x \in \F^n: f_i(x_{tm+1},\ldots,x_{tm+m})=0,\;\forall\; 0 \le t < n/m,\;1 \le i \le k\}.
\end{eqnarray*}
Then $S$ is $(k,(d_1)^k)$-subspace evasive. Moreover,
\begin{enumerate}
\item If $\eps \le 1/10$, $d_1 \le |\F|^{1/4}$ and $|\F|^m \ge n^8$ then $|S| \ge \frac{1}{3}|\F|^{(1-\eps) n}$.
\item If at least $k$ of the degrees $d_1,\ldots,d_m$ are co-prime to $|\F|-1$ then $|S| = |\F|^{(1-\eps) n}$.
\end{enumerate}
\end{thm}

We prove Theorem~\ref{thm-evasive} in the remainder of this section. We first show that $\V_{\F}(f_1,\ldots,f_k)$ has small intersection with affine subspaces of dimension {\em at most} $k$ (this is a stronger statement than the one we proved in Section~\ref{sec-equations} since the dimension of the subspace can be smaller than $k$).

\begin{claim}\label{claim-small-intersection-upto-k}
Let $H \subset \F^m$ be an $r$-dimensional affine subspace for $r \le k$. Then $|\V_{\F}(f_1,\ldots,f_k) \cap H| \le (d_1)^r$.
\end{claim}

\begin{proof}
Note that $\V_{\F}(f_1,\ldots,f_k) \cap H = \V(f_1,\ldots,f_k) \cap H$ since $H \subset \F^m$. We will show that in fact
$|\V(f_1,\ldots,f_r) \cap H| \le (d_1)^r$, from which the claim will follow since $\V(f_1,\ldots,f_k) \subset \V(f_1,\ldots,f_r)$.
Now, since the matrix $A$ is strongly-regular, its restriction to the first $r$ rows is $r$-regular; hence $\V(f_1,\ldots,f_r)$ is everywhere-finite (as an $(n-r)$-dimensional variety) and, by Bezout's Theorem (Theorem~\ref{thm-bezout}), we have $|\V(f_1,\ldots,f_r) \cap H| \le (d_1)^r$.
\end{proof}

We now prove that $S=\V_{\F}(f_1,\ldots,f_k)^{(n/m)}$ is subspace evasive for dimensions up to $k$.

\begin{claim}
Let $H \subset \F^n$ be an $r$-dimensional affine subspace for $r \le k$. Then $|S \cap H| \le (d_1)^r$.
\end{claim}

\begin{proof}
Let $\V_{\F}=\V_{\F}(f_1,\ldots,f_k)$. We prove the claim by induction of the number of blocks $b=n/m$. If $b=1$ then $S=\V_{\F}$ and the claim follows from Claim~\ref{claim-small-intersection-upto-k}. We thus assume that $b>1$. Decompose $H$ as a disjoint union of subspaces based on the restriction to the first $m$ coordinates $x_1,\ldots,x_m$ (i.e. the first block). That is, let $T:=\{(x_1,\ldots,x_m): (x_1,\ldots,x_n) \in H\}$ and for each $t \in T$ let $H_t:=\{(x_1,\ldots,x_n) \in H: (x_1,\ldots,x_m)=t\}$. Thus $H = \cup_{t \in T} H_t$ and we have that
$$
|\V_{\F}^{n/m} \cap H| = \sum_{t \in \V_{\F} \cap T} |\V_{\F}^{(n/m)-1} \cap H_t|.
$$
Now, since $H$ is an affine subspace so is $T$. Let $r'=dim(T)$ where $0 \le r' \le r$. We also have that $H_t$ is an $(r-r')$-dimensional affine subspace for all $t \in T$. Now by Claim~\ref{claim-small-intersection-upto-k} we have that
$|\V_{\F} \cap T| \le (d_1)^{r'}$; and by induction we have that $|\V_{\F}^{(n/m)-1} \cap H_t| \le (d_1)^{r-r'}$ for all $t \in T$. Hence $|\V_{\F}^{n/m} \cap H| \le (d_1)^r$ as claimed.
\end{proof}

We now turn to prove the `Moreover' part of Theorem~\ref{thm-evasive}, namely to lower bound the size of $S$. To do so, it is enough to bound the size of $\V_{\F}(f_1,\ldots,f_k)$ (since $S$ is a product of such sets). We begin with the unrestricted case, where all we assume are some (rather weak) bounds on the size of the field. We refer the reader to Appendix~\ref{sec:preliminaries} for the notations/preliminaries on Fourier analysis and Weil's theorem used in the proof.

\begin{claim}\label{claim-lower-bound-size-vf}
Assume that $\eps \le 1/10$, $d_1 \le |\F|^{1/4}$ and $|\F|^m \ge n^8$. Then $|\V_{\F}(f_1, \ldots,f_k)| \ge (1-1/n) |\F|^{m-k}$. In particular $|S| \ge |\F|^{(1-\eps)n} / 3$.
\end{claim}

\begin{proof}
Let $x=(x_1,\ldots,x_m) \in \F^m$ be chosen uniformly. Our goal is to estimate the probability that $\sum_{j=1}^m A_{i,j} \cdot x_j^{d_j}=0$ for all $i \in [k]$. Equivalently, let $X^{(j)} \in \F^k$ be a random variable defined as $X^{(j)}_i := A_{i,j} \cdot x_j^{d_j}$ and let $X:=X^{(1)}+\ldots+X^{(m)}$. We need to estimate the probability that $X=0^k$.

To this end, we apply Fourier analysis (for definitions see Appendix~\ref{sec:preliminaries}). Assume $\F=\F_{q}$ where $q=p^{\ell}$. The characters of $\F^k$ are given by $\chi_a(x)=\ep(\Tr(\ip{a}{x}))$ for $a \in \F^k$ where $\Tr:\F_q \to \F_p$ is the trace operator. Since $X^{(1)},\ldots,X^{(m)}$ are independent we have that
$$
\widehat{X}(a) = \E[\chi_a(X)] = \E[\chi_a(\sum_{j=1}^m X^{(j)})] = \prod_{j=1}^m \E[\chi_a(X^{(j)})]=\prod_{j=1}^m \widehat{X^{(j)}}(a).
$$
We proceed to estimate the Fourier coefficients of $X^{(j)}$. Let $A^{(j)} \in \F^k$ denote the $j$-th column of $A$. We have that
$$
\widehat{X^{(j)}}(a) = \E_{x_j \in \F}[\ep(\Tr(\ip{a}{A^{(j)}} \cdot x_j^{d_j}))].
$$
Thus, if the inner product of $a$ and $A^{(j)}$ is nonzero, we have by the Weil bound (Theorem~\ref{thm-weil}) that
$$
\left|\widehat{X^{(j)}}(a)\right| \le \frac{d_j-1}{\sqrt{|\F|}} \le |\F|^{-1/4}.
$$
Since we assume $A$ is strongly-regular, for any nonzero $a \in \F^m$ there could by at most $k-1$ columns of $A$ which are orthogonal to $a$; hence we deduce that for any nonzero $a$,
$$
\left|\widehat{X}(a)\right| \le (|\F|^{-1/4})^{(m-k+1)} \le |\F|^{-k-m/8} \le (1/n) \cdot |\F|^{-k}
$$
by our choice of parameters. We now apply these bounds to estimate the probability that $X=0$.
We have that $\Pr[X=0]=|\F|^{-k} \sum_{a} \widehat{X}(a)$ and $\widehat{X}(0)=1$; hence
$$
\left|\Pr[X=0]-\F^{-k}\right| \le \F^{-k} \sum_{a \ne 0} |\widehat{X}(a)| \le (1/n) \cdot |\F|^{-k}.
$$
Thus $|\V_{\F}(f_1,\ldots,f_k)| \ge (1-1/n) \F^{m-k}=(1-1/n)\F^{(1-\eps)m}$ and $|S|=|\V_{\F}(f_1,\ldots,f_k)|^{n/m} \ge (1/e) \cdot |\F|^{(1-\eps) n}$.
\end{proof}

We now prove the second item of the `Moreover' part of Theorem~\ref{thm-evasive} in which we assume that at least $k$ of the degrees $d_1,\ldots,d_n$ are co-prime to $|\F|-1$ and use this extra condition to obtain a precise quantity for the size of $S$. In Section~\ref{sec-explicit} we show that this condition on the degrees is relatively easy to satisfy if we are the ones choosing the field $\F$.

\begin{claim}\label{claim-size-bound-k-coprime}
If at least $k$ of the degrees $d_1,\ldots,d_m$ are co-prime to $|\F^*|=|\F|-1$ then
$|\V_{\F}(f_1,\ldots,f_k)| = |\F|^{m-k}$. This implies $|S| = |\F|^{(1-\eps)n}$.
\end{claim}

\begin{proof}
Let $\V_{\F} = \V_{\F}(f_1,\ldots,f_k)$. Let $d_{j_1},\ldots,d_{j_k}$ be degrees among $d_1,\ldots,d_m$ co-prime to $|\F|-1$ and let $J=\{j_1,\ldots,j_k\}$. We will show that for any setting of $\{x_j: j \notin J\}$ there exists a unique setting of $\{x_j: j \in J\}$ which makes $x \in \V_{\F}$. This will clearly show that $|\V_{\F}|=|\F|^{m-k}$ as claimed.

Substitute $x_j=c_j \in \F$ for all $j \notin J$. We have that $x \in \V_{\F}$ if
$$
\sum_{j \in J} A_{i,j} \cdot x_{j}^{d_j} = -\sum_{j \notin J} A_{i,j} \cdot c_j^{d_j} \qquad \forall i \in [k].
$$
Let $A'$ be the $k \times k$ minor of $A$ given by restricting $A$ to columns in $J$. Let $y=(x_{j_1}^{d_{j_1}},\ldots,x_{j_k}^{d_{j_k}})$ and let $b \in \F^k$ be given by $b_i = -\sum_{j \notin J} A_{i,j} \cdot c_j^{d_j}$. Then $x \in \V_{\F}$ if
$$
A' y = b.
$$
We have that $A'$ is regular since $A$ is strongly regular; hence there exists a unique solution $y \in \F^k$ for the linear system $A' y = b$. We now apply our assumption that each degree $d_{j_i}$ is co-prime to $|\F^*|=|\F|-1$. This implies that raising to the $d_{j_i}$ power in $\F$ is an automorphism of $\F^*$. That is, for each $y_i$ there exists a unique solution to $x_{j_i}^{d_{j_i}}=y_i$ where $x_{j_i} \in \F$.
\end{proof}

%%%%%%%%%%%%%%%%%%%%%%%%%%%%%%%%%%%%%%%%%%%%%%%%%%%%%%%%%%%%%%%%%%%%%%%%%%
\section{Explicitness of the construction}\label{sec-explicit}
%%%%%%%%%%%%%%%%%%%%%%%%%%%%%%%%%%%%%%%%%%%%%%%%%%%%%%%%%%%%%%%%%%%%%%%
In this section we  discuss the explicitness of our construction of  subspace evasive sets. The construction of everywhere-finite varieties accomplished in Theorem~\ref{thm-explicit-efv} is given as the zero set of explicitly defined polynomials.  One can use our construction over any finite field, including $\F=\F_{2^{\ell}}$ which is convenient for applications. The construction requires  an explicit  strongly regular $k \times n$ matrix  $A$. Such a matrix can be easily obtained when $|\F| > n$ by taking $A_{i,j} = \gamma_j^i$ where $\gamma_1,\ldots,\gamma_n \in \F$ are $n$ nonzero distinct elements in $\F$ (this is because each $k \times k$ sub-matrix is a Vandermonde matrix).

\subsection{Efficient encoding of vectors as elements of $S$}
It is trivial to decide in polynomial time if a given point $x \in \F^n$ is in $S$ or not. The first non-trivial issue regarding explicitness is how to sample an element of the set uniformly. More precisely, for an evasive set $S \subset \F^n$ of size $|\F|^r$ we would like to have an efficiently computable bijection $\varphi : \F^r \mapsto S$. This is needed for the list-decoding application (see Section~\ref{sec-listdec}) because we would like to encode messages as strings in $S$ without losing much in the rate of the code and so that we can efficiently recover the original messages from their representation as elements of $S$. We now show how one can sample from the variety $\V_{\F}(f_1,\ldots,f_k)$ efficiently (this is enough since the construction of evasive sets is a Cartesian product of such sets). We show this is simple when at least $k$ of the degrees $d_1,\ldots,d_m$ are co-prime to $|\F|-1$ (we will show below that this condition is easy to obtain).

\begin{claim}\label{cla:sample}
Assume that at least $k$ of the degrees $d_1,\ldots,d_m$ are co-prime to $|\F|-1$. Then there is an easy to compute bijection $\varphi:\F^{m-k} \to \V_{\F}(f_1,\ldots,f_k) \subset \F^m$. Moreover, there are $m-k$ coordinates in the output of $\varphi$ that compute the identity mapping $Id : \F^{m-k} \mapsto \F^{m-k}$.
\end{claim}

\begin{proof}
The proof is similar to the proof of Claim~\ref{claim-size-bound-k-coprime}. Let $d_{j_1},\ldots,d_{j_k}$ be degrees among $d_1,\ldots,d_m$ co-prime to $|\F|-1$ and let $J=\{j_1,\ldots,j_k\}$. We showed in Claim~\ref{claim-size-bound-k-coprime} that for any setting for $\{x_j: j \notin J\}$ there exists a unique setting of $\{x_j: j \in J\}$ which makes $x \in \V_{\F}(f_1,\ldots,f_k)$.
We now show that given this setting, the values of $\{x_j: j \in J\}$ can be found efficiently. Thus taking $\varphi$ to be the identity map from $\F^{m-k}$ to $\F^{[m]\setminus J}$ and completing it uniquely to $\F^m$ will give the required map.
As we showed in Claim~\ref{claim-size-bound-k-coprime} we have $x_{j_i}^{d_{j_i}}=y_j$ where $y=(y_1,\ldots,y_m)$ is a unique solution to a linear system $A' y=b$, where $A',b$ are easy to compute and $A'$ is regular. The value of $y$ can be found by solving a linear system; and the value of $x_{j_i}$ can be retrieved since $x_{j_i}=y_j^{e_j}$ where $e_j$ is the inverse of $d_{j_i}$ modulo $|\F|-1$ (which exists by assumption).
\end{proof}

\subsection{Computing the intersection with a given subspace}
Another important explicitness issue is how to efficiently compute the intersection of a $(k,c)$-subspace evasive set $S \subset \F^n$ with a given affine subspace $H$ of dimension $k$. This question comes up in the list-decoding application when we obtain a subspace (given in some basis) that is supposed to contain all possible decodings of a corrupted codeword and we wish to `filter-out' this subspace to obtain the list of elements in it that are also in $S$. One way of doing this is to go over all elements in $H$ and to check for each whether or not it is in $S$ (in our case by evaluating the $k$ polynomials and checking that they are all zero). Using the specific structure of our construction we can do much better and output the set $S \cap H$ in polynomial time in the size of the intersection.

\begin{claim}\label{cla:compute-intersect}
Let $S \subset \F^n$ be the $(k,c)$-subspace evasive set constructed in Theorem~\ref{thm-evasive} (for some choice of the parameter $m$ and degrees $d_1> \ldots > d_m$). There exists an algorithm  that, given a basis for any affine subspace $H$ of dimension $k$, outputs $S \cap H$ in time polynomial in the output size.
\end{claim}
\begin{proof}
	This follows from powerful algorithms that can solve a system of polynomial equations (over finite fields) in time polynomial in the size of the output, {\em provided that the number of solutions is finite in the algebraic closure} (i.e the `zero-dimensional' case). See for example \cite{Laz92,FGLM93}. In our basic construction of an everywhere finite variety, given as the common zero set of $k$ polynomials $f_1,\ldots,f_k$ in $n$ variables $x_1,\ldots,x_n$, the intersection with a $k$ dimensional affine subspace reduces to solving a system of  $k$ equations in $k$ variables -- simply substitute $x_i = \ell_i(t_1,\ldots,t_k)$, where $H$ is the image of the degree one map $\ell : \F^k \mapsto \F^n$. For the construction of the evasive set (which is the direct product of these simple varieties) we  can use an iterative argument (similar to the proof of Theorem~\ref{thm-evasive}). Recall that in our construction we partitioned the set of coordinates into consecutive blocks of length $m$ -- each containing an independent copy of a the variety $\V(f_1,\ldots,f_k)$.  In the first step we solve a system of equations for the projection of $H$ on the first $m$ coordinates. If the dimension of this projection is $r_1$ then this step will take time polynomial in $(d_1)^{r_1}$ which is the bound on the number of solutions. For every fixing of the first $m$ coordinates to a solution obtained in this step, we reduce the dimension of $H$ by $r_1$ and obtain a new subspace $H'$ on the remaining coordinates. Continuing in the same fashion with $H'$ on the second block we can compute all solutions in time $\poly((d_1)^{r_1})\cdot \poly((d_1)^{r_2})\cdot \cdots \cdot \poly((d_1)^{r_\ell})$, where $r_1+ r_2 + \ldots + r_\ell = k$. This will add up to a total of $\poly((d_1)^k)$ running time, which is polynomial in the number of solutions.
\end{proof}

\subsection{Generating a field with $k$ degrees co-prime to $|\F|-1$} % (fold)
\label{sub:finding_m_degrees_co_prime_to_f_1_}

% subsection finding_m_degrees_co_prime_to_f_1_ (end)

We now address the condition, appearing  in Theorem~\ref{thm-evasive} and in Claim~\ref{cla:sample}, that at least $k$ of the degrees $d_1>\ldots > d_m$ used in the construction are co-prime to $|\F|-1$. We will want to satisfy this condition, while still maintaining am reasonable bound on $d_1$  (which is important since it determines the intersection size with subspaces). For certain fields it may be the case that $|\F|-1$ has many small divisors, in which case $d_1$ might have to be large. However, if one has the freedom of `picking' the field size (as we do in the application to list-decoding) then this problem can be avoided. In essence, we need a (deterministic) way of generating a field $\F$ of size within some specified range and with at least $k$ small integers co-prime to $|\F|-1$. The best bound on the $k$ integers is $O(k)$ which can be obtained, for example, using `safe' primes or primes of the form $2q+1$ for $q$ prime. Since we do not know how to find a safe prime in a specified range (or even to show that infinitely many such primes exist!) we will have to resort to an asymptotically weaker bound as is given by the following claim.

\begin{claim}\label{cla-field}
There exists a constant $C>0$ such that the following holds: There is a deterministic algorithm that, given integer inputs $k,n$ so that $n > k^{C\log\log k}$,  runs in $\poly(n)$ time and returns a prime $p$ and $k$ integers $k^{C\log\log k}> d_1 > d_2 > \ldots > d_k > 1$ such that:
\begin{enumerate}
	\item For all $i \in [k]$, $\gcd(p-1,d_i)=1$.
	\item $n < p \leq n\cdot k^{C\log\log k}$.
\end{enumerate}
\end{claim}
\begin{proof}
Let $K$ be the product of the first $\lceil \log_2( k +1 ) \rceil $ odd primes. By the prime number theorem, $K < k^{C'\log\log k}$ for some constant $C'>0$. Let $K > d_1 > \ldots > d_k$ be $k$ distinct odd divisors of $K$. We will show how to chose the prime $p$ as in the claim using results on the distribution of primes in arithmetic progressions (see \cite{Iwaniec} for more details). Property 1. will follow if our prime $p$ will satisfy the congruence $p = 2 \mod K$. Since $K$ and $2$ are co-prime, we know (see \cite{Iwaniec}) that the number of primes smaller than $x$ satisfying this congruence is asymptotically $\frac{1}{\phi(K)}\cdot \frac{x}{\ln x}$, where $\phi$ is Euler's totient function. From this it follows that there exists a prime in the range $[n,2nK]$ that satisfies the congruence $p=2 \mod K$ and, consequently, $p-1$ is co prime to all the divisors of $K$. Finding this $p$ in time polynomial in $n$ is trivial since we can just try all integers in the range and test them for primality.
\end{proof}

\section{Proof of Theorem~\ref{thm:list-decodable-codes}}\label{sec-listdec}
 In \cite{Guruswami11}, Guruswami considers an explicit family of codes (folded Reed-Solomon codes) that are defined as the image of an explicit mapping
$$
\mathcal{C}:\F^n \to (\F^r)^{m/r}
$$
where $\F=\F_q$ is any finite field of size $q$ and $1 \le m \le q-1$ is a multiple of $r$. Since our application is black-box (and applies to any code that shares the decoding properties listed below, such as the codes from \cite{GuruswamiWang11}) we omit the precise description of the code and refer the reader to \cite{Guruswami11} for more details on the actual definition of $C$.

In this setting, $n$ is the message length, $\Sigma=\F^r$ is the alphabet, $N=m/r$ is the block length and $R=n/m$ is the rate. Let $\eps>0$ be a sufficiently small constant and set $k \approx 1/\eps$ and $r \approx 1/\eps^2$. Guruswami shows (Theorem 7 in \cite{Guruswami11}) that for the above choice of parameters:
\begin{enumerate}
\item The mapping  $\mathcal{C}$ can be computed in polynomial time.
\item There exists a polynomial time algorithm that, given $y \in (\F^r)^{m/r}$, returns a basis to a subspace $H \subset \F^n$ of dimension $k$ which contains all points $x \in \F^n$ whose encoding $\mathcal{C}(x)$ has normalized hamming distance at most $1-R-\eps$ from $y$. (In fact, this algorithm runs in time quadratic in its output length.)
\end{enumerate}

We now describe how to  combine this code $C$ with our construction of subspace evasive sets (Theorem~\ref{thm-evasive}) to obtain a code $C'$ with shorter list size and without loosing too much in the decoding radius. Let $S \subset \F^n$ be a $(k,c=c(k,\eps))$ subspace-evasive set obtained from Theorem~\ref{thm-evasive}. Using Claim~\ref{cla-field}, we can construct a finite field $\F$ of size between $n$ and $O_\eps(n)$ so that the first $k$ degrees $d_1, d_2, \ldots, d_k$ used in the construction are co-prime to $|\F|-1$. From Claim~\ref{cla:sample} we know that there is an efficiently computable bijection $\varphi:\F^{(1-\eps)n} \to S \subset \F^n$. Consider the composed code $\mathcal{C'}:\F^{(1-\eps)n} \to (\F^r)^{m/r}$ defined as
$$
\mathcal{C'}(x) = \mathcal{C}(\varphi(x)).
$$
Let $R'$ denote the rate of $C'$. Then $$R' = (1-\eps)n/m = (1-\eps)\cdot R \geq R - \eps.$$
First, we claim that the composed code $\mathcal{C'}$ can be list-decoded with list size $c(1/\eps,\eps) = (1/\eps)^{O(1/\eps)}$ from a  fraction $$1 - R - \eps \geq 1 -R' - 2\eps$$ of errors. This is since for every $y \in (\F^r)^{m/r}$, the subspace $H \subset \F^n$ returned by the list-decoding algorithm for $\mathcal{C}$ contains at most $c$ messages who lie in $S$.

In order to maintain the efficiency of encoding and list-decoding of $\mathcal{C'}$, we need to guarantee three properties:
\begin{enumerate}
\item[(i)] Encoding: the map $\varphi$ should be computable in polynomial time.
\item[(ii)] Decoding: the inverse map $\varphi^{-1}$ should be computable in polynomial time.
\item[(iii)] List-decoding: for every subspace $H \subset \F^n$ of dimension $s$, we can find in polynomial time the intersection $S \cap H \subset \F^n$.
\end{enumerate}
The first two items are guaranteed by Claim~\ref{cla:sample}, and the third by Claim~\ref{cla:compute-intersect}. Using the property that the decoded given by Guruswami runs in quadratic time we get that the composed code $C'$ can be list-decoded (with the above parameters) in quadratic time (for all constant $\eps >0$). This completes the proof of Theorem~\ref{thm-evasive}.

\section{Acknowledgments} % (fold)
\label{sec:acknowledgments}
The authors would like to thank Venkatesan Guruswami, Avi Wigderson and Amir Yehudayoff for stimulating conversations. Part of the proof of Theorem~\ref{thm-explicit-efv} is inspired by an unpublished construction by Jean Bourgain of affine extractors and we thank him for sharing his result with us.
% section acknowledgments (end)

%%%%%%%%%%%%%%%%%%%%%%%%%%%%%%%%%%%%%%%%%%%%%%%%%%
\bibliographystyle{alpha}
\bibliography{evasive}
%%%%%%%%%%%%%%%%%%%%%%%%%%%%%%%%%%%%%%%%%%%%%%%%%%%

\appendix

%%%%%%%%%%%%%%%%%%%%%%%%%%%%%%%%%%%%%%%%%%%%%%%%%%%%%%%%%%%%
\section{Fourier analysis} % (fold)
\label{sec:preliminaries}

% section preliminaries (end)
%%%%%%%%%%%%%%%%%%%%%%%%%%%%%%%%%%%%%%%%%%%%%%%%%%%%%%%%%%%%

Let $\F$ be a finite field. An additive character (e.g. Fourier basis) of $\F$ is a function $\chi:\F \to \C$ such that $\chi(x+y)=\chi(x)\chi(y)$ for all $x,y \in \F$. The set of characters form an orthogonal basis of functions from $\F$ to $\C$.
Let $\ep(x):=e^{2 \pi i x / p}$. The set of characters of $\F=\F_{p^{\ell}}$ for $p$ prime and $\ell \ge 1$
is given by $\chi_a(x) = \ep(\Tr(a x))$ for $a \in \F_{p^{\ell}}$, where the trace $\Tr:\F_{p^{\ell}} \to \F_p$ is defined as $\Tr(x)=\sum_{i=0}^{e-1} x^{p^i}$. The constant function $1$ is a trivial character for any field; any other character is called non-trivial.
More generally, the characters of the vector space $\F^m$ are given by $\chi_a(x)=\ep(\Tr(\ip{a}{x}))$ where $a=(a_1,\ldots,a_m) \in \F^m$, $x=(x_1,\ldots,x_m) \in \F^m$ and $\ip{a}{x}=a_1 x_1 + \ldots + a_m x_m$.

Let $X$ be a random variable taking values in $\F^m$. Its Fourier coefficients $\widehat{X}(a)$ for $a \in \F^m$ are given by
$$
\widehat{X}(a) := \sum_{x \in \F^m} \Pr[X=x] \chi_a(x);
$$
and for any $x \in \F^m$, the Fourier inversion formula gives that
$$
\Pr[X=x] = |\F|^{-m} \sum_{a \in \F^m} \widehat{X}(a) \overline{\chi_a(x)}.
$$

\paragraph{Character sums}
The following result by Weil \cite{Weil48} (see also \cite{Schmidt}) is a strong tool which gives a bound on the average of a nontrivial character evaluated over the output of a low degree polynomial.
\begin{thm}[Weil]\label{thm-weil}
Let $f(x)$ be a non-constant degree $d$ polynomial over a finite field $\F$. Let $\chi:\F \to \C$ be a nontrivial additive character. Then
$
\left|\E_{x \in \F}[\chi(x)]\right| \le \frac{d-1}{\sqrt{|\F|}}.
$
\end{thm}

\end{document}